\documentclass[conference,final,twocolumn]{IEEEtran}
\usepackage{amsmath,amssymb,amsfonts,mathrsfs,bm}
\usepackage{amstext}
\usepackage{upgreek}
\usepackage{multicol}
\usepackage{indentfirst}
\usepackage{graphicx}
\usepackage{paralist}
\usepackage{cite}
\usepackage{citesort}
\usepackage{booktabs}
\usepackage{multirow}
\newtheorem{proposition}{Proposition}
\newtheorem{proof}{Proof}
\IEEEoverridecommandlockouts

\begin{document}
\title{Energy Costs for Traffic Offloading by Cache-enabled D2D Communications}
\author{
%\IEEEauthorblockN{{Binqiang Chen, Chenyang Yang}} \vspace{0.2cm}
%\IEEEauthorblockA{Beihang University, Beijing, China\\
%Email: \{chenbq,cyyang\}@buaa.edu.cn}%\ \ \ \ \ \ \ \ \ \ \ \ \ \ \ \ \ $\:$ }
%\and \IEEEauthorblockN{{Gang Wang}} \vspace{0.2cm}
%\IEEEauthorblockA{NEC Labs, China \\
%Email: wang\_gang@nec.cn}
\IEEEauthorblockN{{Binqiang Chen and Chenyang Yang}}
\IEEEauthorblockA{Beihang University, Beijing  China
\\
Email: chenbq@buaa.edu.cn,\ %\\
cyyang@buaa.edu.cn}%\ \ \ \ \ \ \ \ \ \ \ \ \ \ \ \ \ $\:$
\thanks{This work was supported by National Natural Science Foundation of China (NSFC) under Grant 61120106002 and National Basic Research Program of China, 973 Program 2012CB316003.}
%
%\thanks{This work was supported in part by....}
}
%wang_gang@nec.cn

\maketitle

\begin{abstract}
Device-to-Device (D2D) communications can offload the traffic and boost the throughput
of cellular networks. By caching files at users, content delivery traffic
can be offloaded via D2D links, if a helper user are willing to send the cached file to the user who requests the file. Yet it is unclear how much energy
needs to be consumed at a helper user to support the traffic offloading. In this paper, we strive to find the minimal energy consumption required at a helper user to  maximize the amount of offloaded traffic. To this end, we introduce a user-centric proactive caching policy that can control the energy cost for a helper user to convey a file,  and then optimize the caching policy to maximize the offloaded traffic. To reduce the energy during transmission, we optimize the transmit power to minimize the energy consumed by a helper to send a file. We analyze the relationship between traffic offloading and energy cost with the optimized
caching policy and transmit power by numerical and simulation results, which demonstrate that a significant amount of traffic  can be offloaded with affordable energy costs.
\end{abstract}

%%%%%%%

\section{Introduction}

Device-to-device (D2D) communications is a promising approach to
offload the traffic and boost the throughput of cellular networks, whose typical user cases include content distribution, gaming and relaying \cite{Survey.D2D}.
%
%We can classify potential services can be offloaded onto D2D links into two types: client-serve (C/S) services, e.g.,
%video-on-demand and file downloads, where the user need to fetch file from remote server via base station (BS) as in
%Fig.\ref{fig.1}.(a); peer-to-peer (P2P) services, e.g, social gaming and cooperative streaming, where two user devices
%exchange information via accessed BS as in Fig.\ref{fig.1}.(b) \cite{Andreev.ICM}. Noticing that users access to P2P
%services may be in close proximity, which presents an excellent opportunity for clients to offload traffic onto D2D
%links, authors in \cite{Andreev.JSAC,Andreev.ICM} investigated the benefits of network-assisted D2D offloading and
%showed that the performance gain can be significant compared to conventional cellular networks. Use cases of D2D
%communications for P2P services have also been specified in \cite{3GPP.D2D}.
%
%\begin{figure}[!htb]
%  \centering
%  % Requires \usepackage{graphicx}
%  \includegraphics[width=0.3\textwidth]{D2D-offload}\\
%  \caption{The description of offloading services}\label{fig.1}
%  \vspace{-0.45cm}
%\end{figure}
%%that the energy consumption for transmitter users can be reduced significantly compared to conventional cellular networks.

The lion's share of cellular traffic is video distribution, which will generate more than $69\%$ of
mobile data traffic by 2019 \cite{CISCO}. Nonetheless, a
large amount of traffic is generated by a few contents. On the other hand, the storage of mobile devices grows rapidly with low
cost. Motivated by these facts, cache-enabled D2D communications was proposed in \cite{Mo.Mag13,Golrezaei.TWC} to offload the traffic of video transmission. Without caching at the devices, the users need to fetch their desired content via base station (BS) from a
remote server. By pre-caching popular files at users during the off-peak time according to their possible interests, the desired file of a user can be transmitted via D2D links by the users in proximity when the actual request arrives. Such a proactive caching policy largely alleviates the burden at the BSs during the peak time. To maximize the traffic offloaded by D2D links, the policy to proactively cache popular files at mobile devices was optimized in \cite{JMY.JSAC}, and a distributed reactive caching mechanism was
designed in \cite{JJJ.JSAC}.

Different from the D2D use cases of supporting peer-to-peer services such as gaming, where the users acting as transmitters are willing to send messages to the destination users \cite{Andreev.ICM}, offloading the content delivery traffic by cache-enabled D2D communications needs the help of other users who
have cached the desired files but are not obligated to act as helpers to transmit the files. Due to the limited battery capacity, a natural question from a helper user in such a network is: ``why should I spend energy of my battery to provide you with faster video download?
\cite{Mo.Mag13}"  This makes the energy consumption at a helper user a big concern in D2D communications with caching. Considering the large
potential in traffic offloading by cache-enabled D2D communications, it is an urgent task to evaluate the energy consumed at a helper user to deliver the files. In particular, characterizing the required energy consumption at a helper to maximize the traffic offloading for the network is an important problem.

In previous research efforts \cite{Mo.Mag13,Golrezaei.TWC,JMY.JSAC,JJJ.JSAC}, the energy costs at helper users are never considered. On one hand,  maximal transmit power  is always used to deliver the files. On the other hand, by dividing the users in a cell into clusters and assuming that only the users
within a cluster can establish D2D links, the optimal caching policy  proposed in
\cite{JMY.JSAC} cannot maximize the offloaded traffic in a cell, and the energy cost for a helper is high. This is because
when a user is only allowed to communicate with the users within its cluster, it can not fetch the file from the nearest helper in adjacent cluster who can convey the file with low transmit power.

In this paper, we strive to investigate the energy cost at a helper user spent to maximize the offloaded  traffic. Toward this goal, we first introduce a user-centric proactive caching policy, where  only the users within a \emph{collaboration distance} of a user can serve as helpers. We optimize the caching policy to maximize traffic
offloading with a given collaboration distance and the user demands statistics. When the collaboration distance is large, the probability that the users can fetch their desired contents via D2D links is high, and then more traffic can be offloaded. However, since the possible D2D link distance increases, the energy cost of a helper user also grows. By setting the collaboration distance according to the affordable energy cost of each helper, such a caching policy helps reduce the energy for transmission. Then, we optimize the transmit power for a helper to send the requested file to minimize the energy spent by the helper. Finally, we analyze the tradeoff between traffic offloading
and energy costs for the D2D network with optimized caching policy and transmit power with numerical and simulation results.

\section{System Model}

Consider one cell in cellular networks where user location follows a
Poisson Point Process (PPP) with density $\lambda$.
%, which is denoted as  $\Phi_u$.
Each single antenna user has local cache to store files and acts as helper. If a helper establishes a D2D link with a D2D receiver (DR), it becomes a D2D transmitter (DT). For notational
simplicity, we assume that each user stores one file as in \cite{Golrezaei.TWC}. The BS is aware of the cached files
of the users and coordinates the D2D communications.
% We assume that the users located in different cells
%can establish D2D links, where the D2D link discovery and establish techniques can be found in \cite{Andreev.ICM,3GPP.D2D}%%
%as illustrated in Fig. \ref{fig.ici-d2d},
%\begin{figure}[!htb]
%  \centering
%  % Requires \usepackage{graphicx}
%  \includegraphics[width=0.3\textwidth]{D2D-model}\\
%  \caption{The D2D communications model}\label{fig.ici-d2d}
%  \vspace{-0.45cm}
%\end{figure}
%\subsection{Content Model}

We consider a static content catalog including $N_f$ files that all users in the cell may request, where the
files are indexed according to the popularity, and the 1st file is the most popular file. Each file is with size of
$F$ Bits. Each user  requests a file from the catalog  independently. The probability that the $i$th file is
requested by a user is assumed to follow Zipf distribution, which is
\begin{equation}
\label{equ.p_r}
p_r(i)=i^{-\beta}/\sum_{k=1}^{N_f}k^{-\beta},
\end{equation}
where $\sum_{i=1}^{N_f}p_r(i)=1$, and the parameter $\beta$ reflects the popularity of the files \cite{Zipf99}.

%\subsection{Communication Model}

\begin{figure}[!htb]
  \centering
  % Requires \usepackage{graphicx}
  \includegraphics[width=0.25\textwidth]{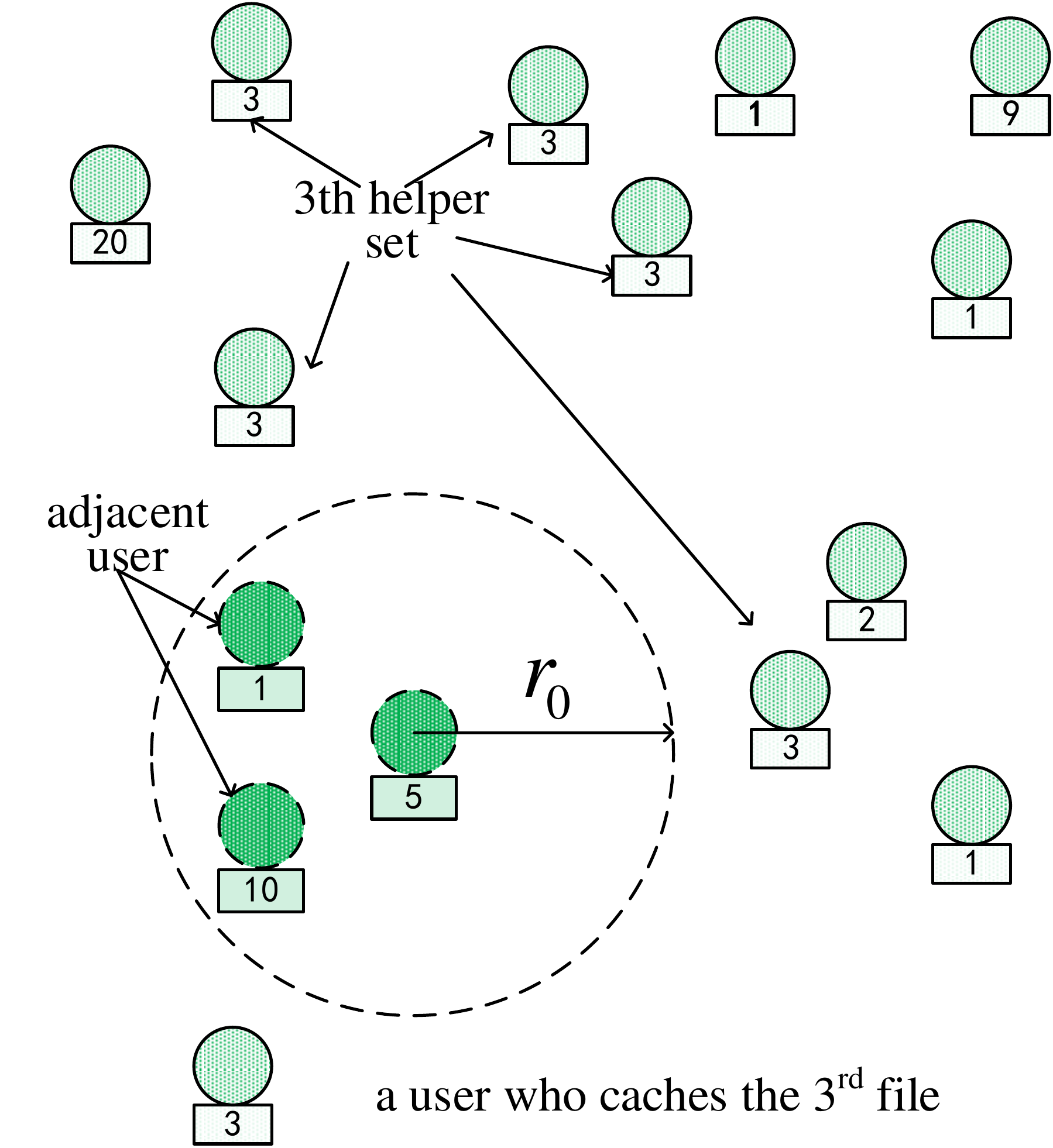}\\
  \caption{Illustration for a cache-enabled D2D network}\label{fig.2}
  \vspace{-0.50cm}
\end{figure}

Since transmitting to a far away DR spends more energy of a DT, we introduce a user-centric caching policy to control the energy for transmission: a helper user will send the file it cached to the user requesting the file only if their distance is smaller than a given value
$r_{c}$, called \emph{collaboration distance}. We consider a probabilistic caching policy, where each user caches a
file according to a $r_c$-dependent caching distribution, i.e., the probability that the $i$th file is  cached at users, $i=1,\cdots,N_f$. All users in the cell
cached with the $i$th file constitute a user set, called the $i$th \emph{helper set}. In practice,
the files can be proactively downloaded by the BS during the off-peak time.

The  users with distance $r$ less than $r_c$  are called \emph{adjacent
users}, as shown in Fig. \ref{fig.2}. If a user can find its requested file in the local caches of its adjacent users,
a D2D link will be established between the user and its nearest adjacent user cached with the requested file to convey
the file. Otherwise, the user needs to fetch the file from the BS. The probability that the requests from the users
can be served via D2D links is called traffic \emph{offloading ratio}, which reflects how much traffic can be
offloaded by D2D communications.

%\subsection{Caching placement}

\section{Optimal Caching policy}
In this section, we optimize the caching distribution to maximize the offloading ratio.
%Actually, by dividing the cell into several clusters and only allowing D2D communications within each cluster, authors
%in \cite{JMY.JSAC} has proposed an optimal caching distribution with given number of users each cluster. However, such
%simplified assumptions have some limitations, e.g., a DR may not access to the nearest DT in its corresponding
%\emph{content dock}, and is not applicable to our scenario. Therefore, the general optimal caching distribution need
%to be obtained.

Denote the probability that the $i$th file is cached at users as $p_c(i)$. Then, the locations of the users who belong to  the $i$th \emph{helper set} follow a PPP distribution with
density $\lambda p_c(i)$ \cite{SKM.PPP}. Thus, the probability that a user requesting the $i$th file can
find its desired file in the cache of any user within the collaboration
distance $r_{c}$ is
\begin{equation} \label{equ.p_f}
\begin{split}
p_f(i) =1-e ^{-\lambda p_c(i) \pi r_{c}^2}.
\end{split}
\end{equation}
Then, the offloading ratio with given caching distribution and collaboration
distance  can be obtained  from (\ref{equ.p_r}) and
(\ref{equ.p_f}) as
\begin{equation} \label{equ.p_o}
{p}_o(p_c(i), r_{c})=\sum_{i=1}^{N_f}p_r(i)(1-e ^{-\lambda p_c(i) \pi r_{c}^2}).
\end{equation}

The optimal caching distribution that maximizes the offloading ratio can be found from the following problem
\begin{equation}
\label{equ.opt1}
\begin{aligned}
&\max_{p_c(i)} \,\, {p}_o(p_c(i), r_{c})\\
&s.t.\quad
\sum_{i=1}^{N_f} p_c(i)=1, \quad p_c(i) \geq 0, \quad i=1,\cdots, N_f.
\end{aligned}
\end{equation}

Because the objective function is the sum of $N_f$ exponential functions with linear constraints, this problem is
convex  \cite{SL.OPT}.
%The Lagrangian function for the problem is
%\begin{equation} \label{equ.Lag_1}
%\Lambda(\mathbf{p_c},\mu) =\sum_{i=1}^{N_f}p_r(i)(1-e ^{-\lambda p_c(i) \pi r_{c}^2}) + \mu (\sum_{i=1}^{N_f} p_c(i)-1),
%\end{equation}
%where $\mathbf{p_c} = [p_c(1),p_c(2),...p_c(N_f)]$ and $\mu$ is the Lagrangian multiplier. By taking the partial derivative with respect to
%$p_c(i)$, we obtain
%\begin{equation} \label{equ.Lag_2}
%p_r(i)\lambda \pi r_{c}^2e ^{-\lambda p_c(i) \pi r_{c}^2} + \mu = 0,
%\end{equation}
According to the Karush-Kuhn-Tucker (KKT) conditions of this problem, the optimal caching distribution should satisfy the
following conditions
\begin{equation} \label{equ.Lag_3}
\begin{split}
& p^*_c(i) = \left[\frac{1}{\lambda \pi r_{c}^2}\ln(p_r(i))-\frac{1}{\lambda \pi r_{c}^2}\ln(\frac{-\mu}{\pi \lambda
r_{c}^2})\right]^+, \forall i, \\
& \sum_{i=1}^{N_f} p^*_c(i) =1,
\end{split}
\end{equation}
where  $[x]^+ = \max(x,0)$.

\begin{proposition} \label{p:1}
If $\frac{(N_f)^{N_f}}{N_f!} < e^{\frac{\lambda \pi r_{c}^2}{\beta}}$, then the optimal caching distribution will be
\begin{equation} \label{p:1.1}
p^*_c(i) = \frac{\beta}{\lambda \pi r_{c}^2 N_f}\sum_{j=1}^{N_f}\ln(\frac{j}{i})+\frac{1}{N_f}.
\end{equation}
Otherwise, the optimal caching distribution will be
\begin{eqnarray}
\label{equ.p:1.2}
p^*_c(i)=
\begin{cases}
\frac{\beta}{\lambda \pi r_{c}^2 i^*}\sum_{j=1}^{i^*}\ln(\frac{j}{i})+\frac{1}{i^*}, &i \leq i^* \\
0, &i^*<i \leq N_f\\
\end{cases},
\end{eqnarray}
where $i^*$ satisfies $\frac{(i^*+1)^{i^*}}{i^*!} \geq e^{\frac{\lambda \pi r_{c}^2}{\beta}}$,
$\frac{(i^*)^{i^*}}{i^*!} < e^{\frac{\lambda \pi r_{c}^2}{\beta}}$ and $\frac{\lambda \pi r_{c}^2}{\beta} -1 \leq i^*
\leq \frac{\lambda \pi r_{c}^2}{\beta} + \ln(\sqrt{2\pi N_f}) + 1$.
\end{proposition}

\begin{proof}
See Appendix A
\end{proof}

%\begin{proposition} \label{p:2}
%The $i^*$ in Prop. \ref{p:1} can be obtained in the range $\frac{\lambda \pi r_{c}^2}{\beta} -1 \leq i^* \leq \frac{\lambda \pi r_{c}^2}{\beta} +
%\ln(\sqrt{2\pi N_f}) + 1$. When $r_{c}\rightarrow \infty$, $i^* \rightarrow \frac{\lambda \pi r_{c}^2}{\beta}  $.
%\end{proposition}

With Prop. \ref{p:1}, the optimal caching distribution  $p^*_c(i), i=1,\cdots, N_f$ can be obtained efficiently, which depends on the collaboration
distance $r_c$, user density $\lambda$, as well as content statistics $N_f$ and $\beta$. With the optimized caching distribution, each user can randomly select a file to cache according to the probability $p^*_c(i)$. Then, the traffic in the network can be maximally offloaded. Because such a caching policy depends on $r_c$, the possible D2D link distance and hence the energy spent by a helper can be controlled by the pre-determined value of collaboration
distance.

\section{Energy Consumption for  Traffic Offloading}
In this section, we investigate the energy consumed by each DT for offloading the traffic by the cache-enabled D2D communications. To this end, we first optimize the transmit power  for conveying a file to minimize the energy consumption at each DT. Then, we evaluate  the average energy consumed at each DT to transmit a file via D2D links based on the optimal transmit power and optimal
caching policy.
%\subsection{Average user data rate}

Once a D2D link is established, the DT can transmit its cached file to the DR who requests the file. Considering the
random user requests, random caching and random user locations, it is reasonable to treat the interference among the D2D links
as noise. Then,
 the signal to interference plus noise ratio (SINR) at the DR can be expressed as
\begin{equation} \label{equ.gamma_DR}
\gamma(r) = \frac{P_t|h|^2r^{-\alpha}}{\sigma_0^2},
\end{equation}
where $P_t$ is the transmit power at the DT, $\sigma_0^2$ is the variance of white
Gaussian noise and interference, $h$ and $r$ are respectively the channel coefficient and distance
between the DT and the DR with $h$ following  zero mean complex Gaussian distribution with unit variance, and
$\alpha$ is the path loss exponent.

Considering that $1+\gamma(r)$ approximately follows Gamma distribution \cite{Gamma.2011}, the average data rate with respect to small scale channel fading can be derived as
\begin{equation} \label{equ.gamma_DR2}
\bar{R}(r) =\mathbb{E} \{W\log_2(1+\frac{P_t|h|^2r^{-\alpha}}{\sigma_0^2})\} \approx W\log_2(1+\frac{P_tr^{-\alpha}}{\sigma_0^2}),
\end{equation}
where $W$ is the bandwidth.

By using the first order approximation \cite{approx}, the average time to convey a file of size $F$ can be approximated as
$F/\bar{R}(r)$. Then, the energy consumed  to transmit a file via a D2D link with distance $r$ can be approximated as
\begin{equation} \label{ECr}
E_c(r) \approx \frac{F}{W\log_2(1+\frac{P_t(r)r^{-\alpha}}{\sigma_0^2})}\left(\frac{1}{\eta}P_t(r)+P_c\right),
\end{equation}
where $P_c$ is the circuit power consumed at the DT, and  $\eta$ is the power amplifier efficiency.
%The power consumption in a real transmission chain depends on various factors such as RF power amplifier efficiency
%and associated circuit blocks. To have a realistic but also analytically tractable power model, we assume that the
%power consumed by the power amplifier is linearly dependent on output power of power amplifier, i.e., constant drain
%efficiency \cite{KV.TWC}. To evaluate the energy consumption during transmission process at a DT, we take into circuit
%power and transmit power into account,
%\begin{equation} \label{equ.P_c}
%P_{all} = \frac{1}{\eta}P_t + P_c,
%\end{equation}
%where is  at DT

To minimize the energy consumption, the optimal transmit power at the DT can be
obtained from the following problem,
\begin{equation}
\label{equ.opt2}
\begin{aligned}
&\min_{P_t(r)}  \,\, \frac{F}{W\log_2(1+\frac{P_t(r)r^{-\alpha}}{\sigma_0^2})}\left(\frac{1}{\eta}P_t(r)+P_c\right)\\
&s.t.\quad
0<P_t(r)\leq P_{\max},
\end{aligned}
\end{equation}
where $P_{\max}$ is the maximal transmit power of the DT.

\begin{proposition} \label{p:2}
Denote $y=1+P_t(r)\frac{r^{-\alpha}}{\sigma_0^2}$, $y_0=
1+P_{max}\frac{r^{-\alpha}}{\sigma_0^2}$, and $\epsilon = \frac{r^{-\alpha}\eta P_c}{\sigma_0^2}-1$. If $(\frac{y_0}{e})^{y_0} < 2^{{\varepsilon}/{\ln2}}$, then the optimal solution of problem  (\ref{equ.opt2}) will be
$P^*_t(r) = P_{\max}$. Otherwise, $P^*_t(r)$ will satisfy the following condition
\begin{equation}
(\frac{y}{e})^y = 2^{{\varepsilon}/{\ln2}}.
\end{equation}
\end{proposition}
\begin{proof}
See Appendix B.
\end{proof}

With Prop. 2, the optimal transmit power $P^*_t(r)$ can be obtained efficiently.
By substituting $P^*_t(r)$ to \eqref{ECr}, the minimal energy consumption for the D2D link with distance $r$ can be obtained as $E_c^*(r)$.

%\subsection{Average energy consumption}
To evaluate minimal energy cost for a DT to transmit a file over D2D links with different distances, we need to obtain the distribution of $r$ when the optimal caching policy is employed. Similar to (\ref{equ.p_f}), the cumulative
distribution function of the distance between the DR requesting the $i$th file and its nearest DT in the $i$th
\emph{helper set} with optimized caching policy can be obtained as $F(p^*_c(i),r)=1-e ^{-\lambda p^*_c(i) \pi r^2}$.
Therefore, the probability density function of the D2D link distance can be obtained  as
\begin{equation} \label{equ.pdf}
f(p^*_c(i),r)=\frac{\text{d}F(p^*_c(i),r)}{\text{d}r}=2\pi r \lambda p^*_c(i) e^{-\lambda p^*_c(i) \pi r^2}.
\end{equation}
Then, for a given collaboration distance $r_c$, the average energy consumed at the DT with the optimized transmit power can be obtained  as
\begin{equation} \label{equ.E_cost}
\begin{split}
&\bar{E}^*_c(r_c) = \sum_{i=1}^{N_f} p_r(i) \int_{0}^{r_{c}} E^*_c(r)f(p^*_c(i),r)dr \\
& = \int_{0}^{r_{c}} E^*_c(r) {p}_o'(p^*_c(i),r_c)dr
\end{split}
\end{equation}
where ${p}_o'(p^*_c(i),r_c)$ is the first-order derivative of  $ {p}_o(p^*_c(i),r_c)$ with respect to $r_c$, and $ {p}_o(p^*_c(i),r_c)$ is obtained from \eqref{equ.p_o} with the optimized caching distribution.

%\subsection{Tradeoff relationship}
%With given caching and power control policy, we can establish the tradeoff relationship between the offloading ratio
%and energy costs via collaboration distance $r_{c}$.
%%
%Denote $\bar{E}^*_c(r_{c})$ as the average energy
%costs with the optimal transmit power.
%Then, we call  $(\bar{E}^*_c(r_{c}),\bar{p}_o^*(r_{c}))$  an offloading-cost
%pair with $r_{c}$.

\vspace{2mm}\begin{proposition} \label{p:3}
Both the maximal offloading ratio $ {p}_o(p^*_c(i),r_c)$ and the minimal average energy cost $\bar{E}^*_c(r_{c})$
increase with the  collaboration distance $r_{c}$.

\begin{proof}
See Appendix C.
\end{proof}

\vspace{2mm}Prop. 3 implies that there is a tradeoff between the traffic offloading and the energy
consumption.

%From Prop. \ref{p:3}, we can establish the tradeoff relationship between traffic offloading ratio and energy costs as
%\begin{equation}
%%\begin{split}
%\bar{E}_c(r_c) = \int_{0}^{r_{c}} E^*_c(r)\bar{p}_o'(\mathcal{P}^*_c(r_c),r)dr
%%\end{split}
%\end{equation}
%where $\bar{p}_o'(\mathcal{P}^*_c(r_c),r) = \frac{\text{d}\bar{p}_o(\mathcal{P}^*_c(r_c),r)}{\text{d}r}$.
%If the traffic load becomes higher, resulting in more traffic offloading ratio is required, then we must increase the
%collaboration distance, while more energy costs are consumed at users. On the other hand, if the energy costs due to
%D2D communications at users take up much of the battery capacity, the collaboration distance should be reduced,
%leading to decrease of offloading ratio.

\end{proposition}

\section{Simulation and Numerical Results}
In this section, we evaluate the accuracy of the approximations and the energy consumption at a DT  for offloading via simulation and numerical results.

To show the impact of the optimized caching policy and transmit power, we consider uniform
caching policy (i.e., all users select a file from the catalog uniformly) as a \emph{caching baseline}, and a transmit policy always using the  maximal transmit power  as the \emph{transmission baseline}.

We consider a square cell with side length $1000$m. The users' location  follows PPP distribution with $\lambda = 0.03$, then there are $2
\sim 3$ users in a $10\text{m}\times 10\text{m}$ area. The
path-loss model is $37.6+36.8\log_{10}(r)$, where $r$ is the distance of D2D link \cite{JMY.JSAC}. $W=20$ MHz and
$\sigma_0^2 = -95 $ dBm. The file catalog is with $N_f=1000$ files, where each file is with size of $30$ Mbytes (i.e.,
a typical video size on the Youtube website \cite{Mo.Mag13}). The parameter of Zipf distribution $\beta=1$. The
maximal transmit power at DT is $P_{\max}=23 $ dBm, the amplifier efficiency $\eta = 0.2$, and the circuit power is
$P_c=115.9 $ mW \cite{KV.TWC}.

To reflect how much average energy consumed at a DT to transmit a file occupies the battery capacity, we
evaluate an \emph{energy cost ratio} as follows,
\begin{equation} \label{equ.E_ratio}
\bar{p}_E = \frac{\bar{E}_c(r_c)}{3.6V_0Q},
\end{equation}
where  $Q$ is the  battery capacity in mAh, and $V_0$ is  the operating voltage in V.

The operating voltage at user is set as $V_0 = 4V$ and the battery capacity $Q = 1800 $ mAh
(typical for iPhone 6).

%\subsection{Optimal Caching Distribution and Transmit Power}
\begin{figure}[!htb]
  \centering
  % Requires \usepackage{graphicx}
  \includegraphics[width=0.45\textwidth]{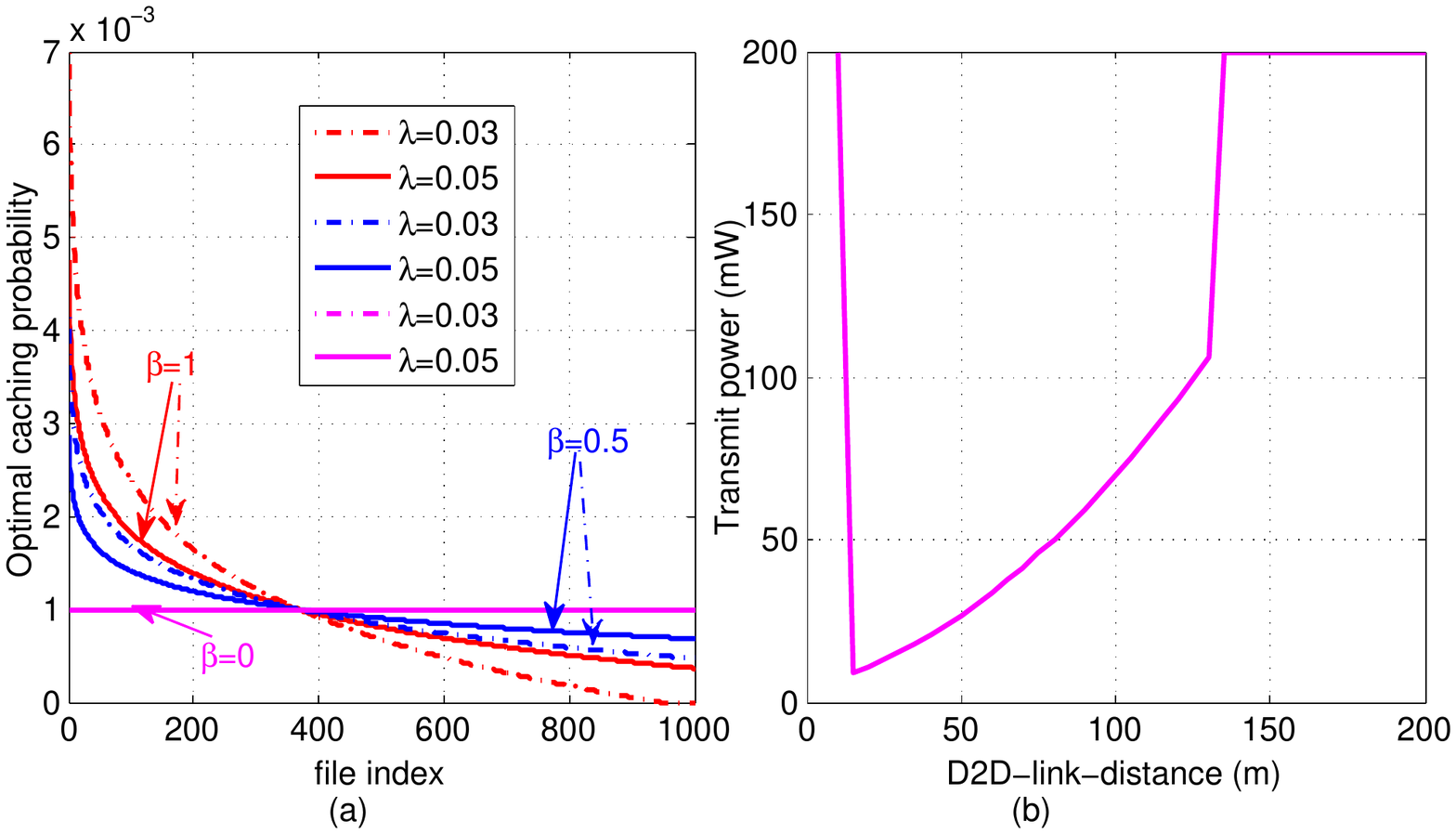}\\
  \caption{ Optimal caching distribution and optimal transmit power}\label{fig.3}
  \vspace{-0.45cm}
\end{figure}
In Fig. \ref{fig.3}, we show the optimal caching distribution and the optimal transmit power versus the D2D link distance $r$. We can see from Fig. \ref{fig.3}(a) that the optimal caching distribution is similar
to the Zipf distribution, where the file with smaller index (i.e., more popular) has higher probability to be cached, which
agrees with the intuition. With the increase of $\beta$ and $\lambda$, the caching probability for popular files increases. We can see from Fig. \ref{fig.3}(b) that the optimal
transmit power for very small and large D2D link distance $r$ is the maximal transmit power. This is because when $r$
is small, the transmit duration can be reduced and the energy cost can be minimized with $P_{\max}$.
When $r$ is large, the circuit power and transmit power are balanced to minimize the energy cost with $P_{\max}$.
\begin{figure}[!htb]
  \centering
  % Requires \usepackage{graphicx}
  \includegraphics[width=0.45\textwidth]{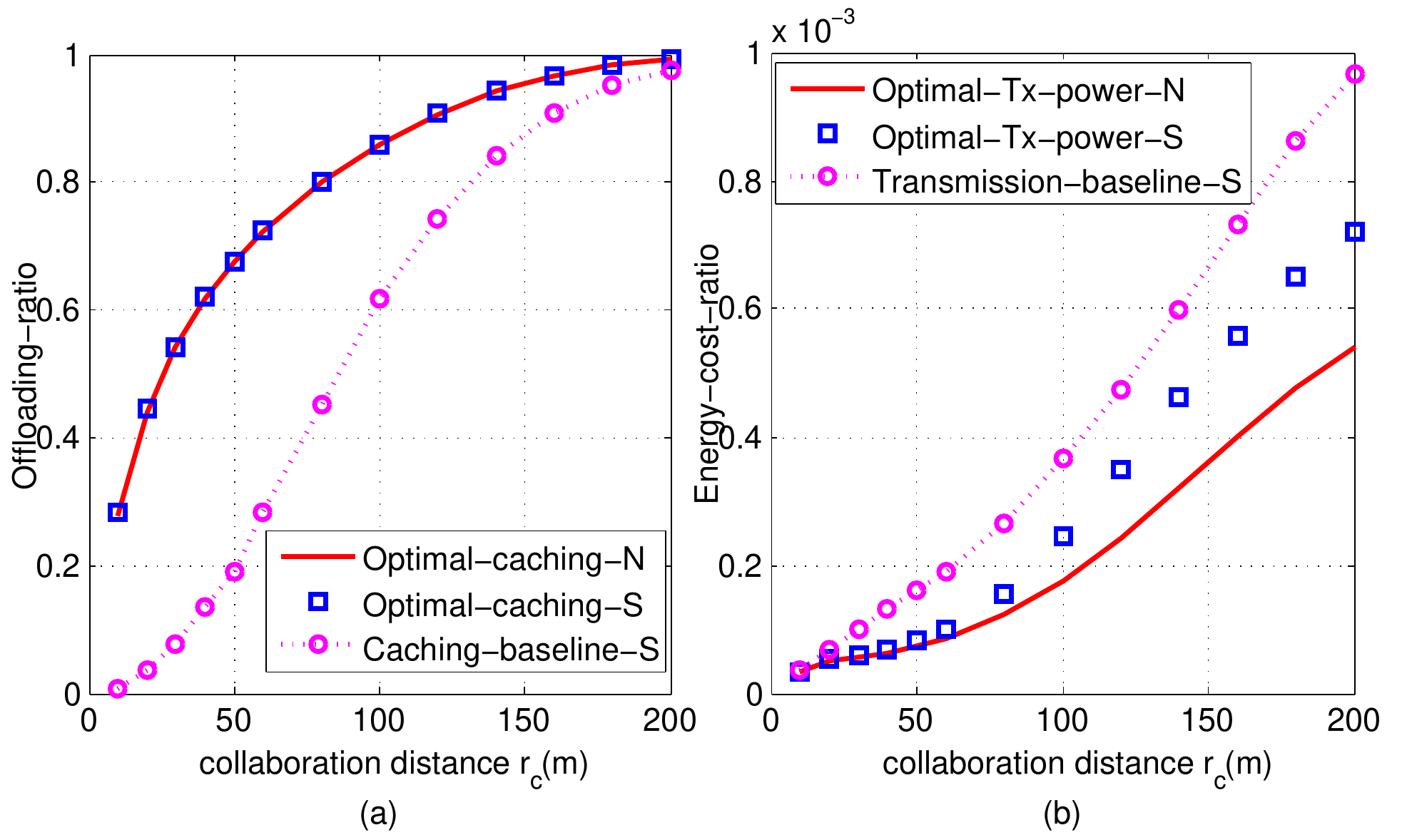}\\
  \caption{ Offloading ratio and average energy costs versus $r_c$, $\beta=1$ $\lambda=0.03$, S-Simulation results, N-Numerical results}\label{fig.4}
  \vspace{-0.45cm}
\end{figure}

In Fig. \ref{fig.4}, we show the offloading ratio and the energy cost ratio versus the collaboration distance $r_c$, where the energy consumption at each DT is computed for the cache-enabled D2D communications with optimized caching distribution. The simulation results
are close to numerical results when $r_c<100$ m, which indicates that the approximated energy consumption is accurate for smaller collaboration distance. As expected, the
optimized caching policy can offload more traffic than the baseline. Even when $r_c=10$ m, more than $20\%$ traffic
can be offloaded with the optimal caching policy. When $r_c$ is large, e.g., $r_c=200$ m, the offloading ratios of both the optimal
and baseline caching policy approach to one, since adjacent users for each user is abundant to cache almost all the
files in the catalog.
The optimized transmit power consumes less energy than always using the maximal transmit power, especially
when the collaboration distance is large. From this figure, we can pre-determine the collaboration distance to control the energy cost at each helper into an affordable level.

\begin{figure}[!htb]
  \centering
  % Requires \usepackage{graphicx}
  \includegraphics[width=0.45\textwidth]{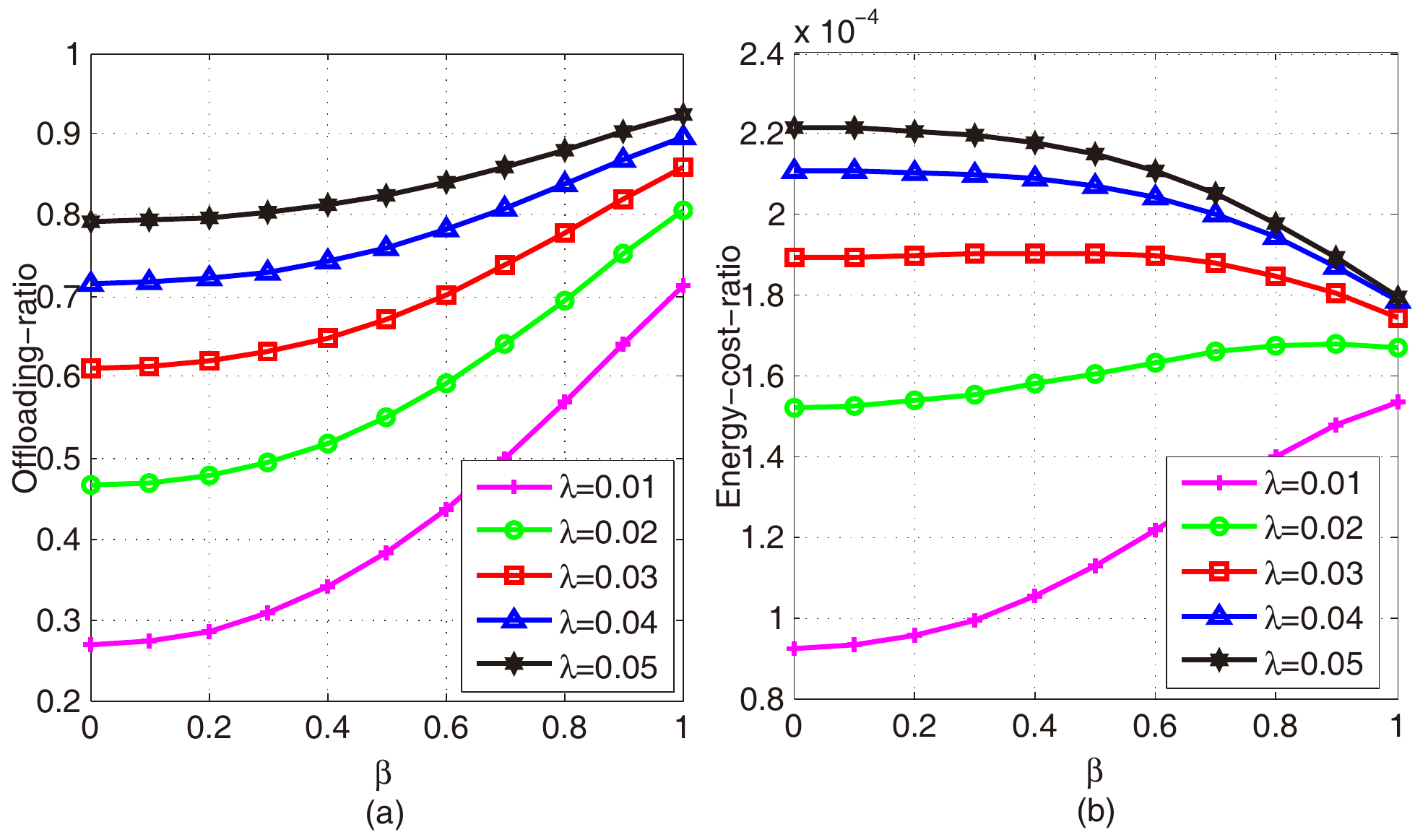}\\
  \caption{ Offloading ratio and average energy costs versus $\beta$ and $\lambda$, $r_c=100$m}\label{fig.5}
  \vspace{-0.45cm}
\end{figure}

In Fig. \ref{fig.5}, we show the offloading ratio with optimal caching policy and the energy cost ratio with optimal transmit power versus Zipf distribution parameter
$\beta$ and user density $\lambda$. With the increase of $\lambda$, a user has more adjacent users, therefore, the probability that a requested
file can be transmitted via D2D link increases.
As a result, the  energy cost ratio increases. The offloading ratio increases with $\beta$, and grows rapidly for heavy traffic load. However,
with the growth of $\beta$, the energy cost ratio increases when $\lambda$ is small but decreases when $\lambda$ is large.
%This comes from two factors: when $\beta$ is large, the high offloading ratio may also cause the
%energy cost ratio, while the distance of D2D link is small that may reduece the energy cost ratio.

%\subsection{Tradeoff relationship}
In Fig. \ref{fig.6}(a), we show the tradeoff between the offloading ratio and the energy cost ratio with
different caching and power control policies, where the offloading ratio is adjusted by changing the collaboration distance $r_{c}$ from $10$m to $100$m. We can see that with the optimized caching policy and transmit power,
the energy cost ratio increases with offloading ratio slowly, and is much lower than other policies.
To offload $80\%$ traffic, the energy
consumption at each DT only occupies $0.02\%$ battery capacity. This suggests that D2D communications with caching will be cost-efficient for offloading if the caching and transmission policies are judiciously designed. 

In the sequel, we show the impact of several key factors on the tradeoff between the offloading ratio and the energy cost ratio with the optimized caching policy and transmit power.
\begin{figure}[!htb]
  \centering
  % Requires \usepackage{graphicx}
  \includegraphics[width=0.5\textwidth]{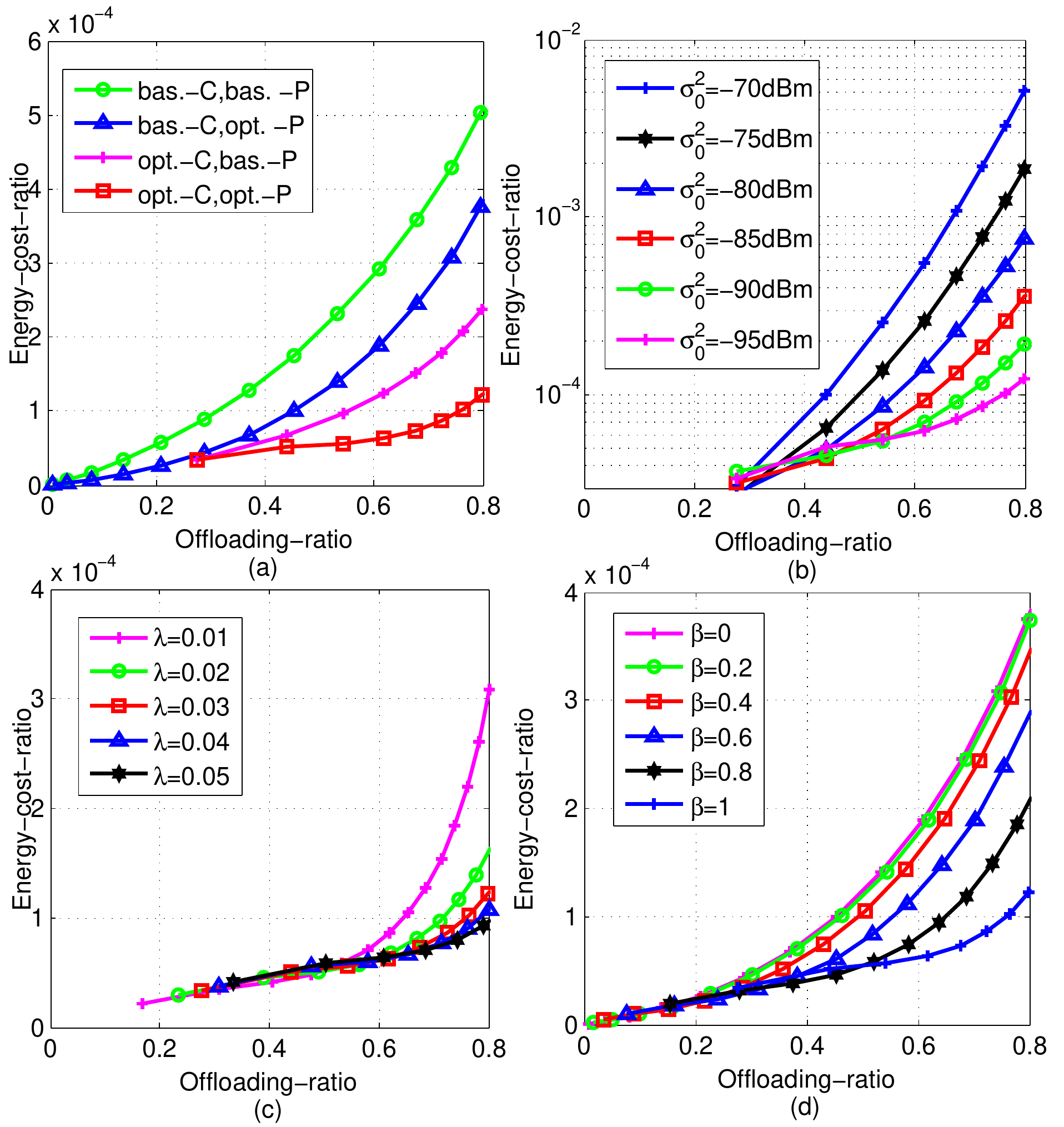}\\
  \caption{Tradeoff  between offloading ratio and energy cost ratio,
  ``opt.'': optimal, ``bas.'': baseline, ``C'': caching, ``P'': power control, where $\lambda=0.03$, $\beta=1$, $\sigma_0^2=-95$dBm if not specified.}\label{fig.6}
  \vspace{-0.45cm}
\end{figure}

In Fig. \ref{fig.6}(b), we show the impact of interference level. As expected, with
the increase of interference level, to achieve the same offloading ratio, more energy needs to be consumed at each DT. When $\sigma_0^2=-70$ dBm, to achieve an $80\%$ offloading ratio, the
energy cost ratio is $1\%$. This suggests the importance of controlling the interference among D2D links.

In Fig. \ref{fig.6}(c), we show the impact of  the users density $\lambda$.  As $\lambda$ increases,
the number of adjacent users of each user increases, hence less energy cost is consumed to achieve the same offloading ratio,
resulting in the reduction of the energy cost ratio. Moreover, the increasing
speed of energy cost with the offloading ratio becomes slowly when $\lambda$ is large. This suggests that it is more efficient in terms of the energy cost of each DT for offloading by cache-enabled D2D communications in the network with heavy traffic load.

In Fig. \ref{fig.6}(d), we show the impact of Zipf distribution $\beta$. As expected,
with the increase of $\beta$, less energy is consumed to achieve a given offloading ratio. To achieve an
$80\%$ traffic offloading ratio, the energy cost of each DT for the scenario with $\beta=1$ is only $25\%$ of that with $\beta=0$.

\section{Conclusion}
In this paper, we strive to evaluate the energy consumed at a helper user to support  traffic offloading for a network by cache-enabled D2D
communications. We first optimized a user-centric proactive caching policy, with which the traffic can be maximally offloaded and the energy consumed for transmission can be controlled by a collaboration distance. We then optimized the transmit power to convey a file via D2D link, aimed to minimize the energy consumption at the helper user for any given caching policy. With the optimal caching distribution and optimal transmit power, we investigated the
tradeoff  between traffic offloading and energy cost. Simulation and numerical results showed that the traffic can be significantly offloaded by the cache-enabled D2D links with
little energy costs at each help user.

\vspace{-0.30cm}
\section*{Appendix $A$: Proof of Proposition \ref{p:1}}
Denote
\begin{equation} \label{xi}
x_i \triangleq \frac{\ln(p_r(i))}{\lambda \pi r_{c}^2},~~ \text{and}\quad
 v \triangleq \frac{1}{\lambda \pi
r_{c}^2}\ln(\frac{-\mu}{\pi \lambda r_{c}^2}).
\end{equation}
Then, by substituting the first condition
into the second condition, the necessary condition  in  (\ref{equ.Lag_3}) can be rewritten as
\begin{equation} \label{equ.proof_1.1}
\sum_{i=1}^{N_f} [x_i-v]^+ =1.
\end{equation}

Since problem \eqref{equ.opt1} is convex, the solution of $v$  found from this condition is global  optimal, with which the optimal caching distribution can be obtained from  (\ref{equ.Lag_3}).

Because $p_r(i)$ is an decreasing function of file index $i$  from (\ref{equ.p_r}), and $p^*_c(i)$ decreases with $p_r(i)$ from
\eqref{equ.Lag_3}, $p^*_c(i)$ is an decreasing function of $i$. Thus, there exists a unique  file index $i^*\leq N_f$, where
$p^*_c(i)>0$ when $i \le i^*$, $p^*_c(i)= 0$ otherwise. Finding the solution of $v$ from  \eqref{equ.proof_1.1} is equivalent to finding the index $i^*$ from
\begin{equation} \label{equ.proof_1.2}
\sum_{i=1}^{i^*} (x_i-v) =1,
\end{equation}
since once $i^*$ is found, $v^*$ can be obtained as
\begin{equation} \label{equ.proof_1.3}
v^* = \frac{\sum_{i=1}^{i^*}x_i - 1}{i^*}.
\end{equation}
\begin{itemize}
\item Case 1: When $i^*=N_f$, $p^*_c(N_f)=x_{N_f}-v > 0$, which can be rewritten as $\sum_{i=1}^{N_f}(x_i - x_{N_f})
    < 1$ after substituting $v$ in (\ref{equ.proof_1.3}). Considering \eqref{xi} and (\ref{equ.p_r}), we have
\begin{equation} \label{equ.proof_1.4}
\begin{split}
&\sum_{i=1}^{N_f}(x_i - x_{N_f}) = \sum_{i=1}^{N_f} \frac{\ln(p_r(i))-\ln(p_r(N_f))}{\lambda \pi r_{c}^2}\\
& {=} \frac{\beta}{\lambda \pi r_{c}^2} \sum_{i=1}^{N_f} \ln(N_f/i) = \frac{\beta}{\lambda \pi r_{c}^2}\ln(\frac{N_f^{N_f}}{N_f!})< 1,
\end{split}
\end{equation}
which can be rewritten as $\frac{(N_f)^{N_f}}{N_f!} < e^{\frac{\lambda \pi r_{c}^2}{\beta}}$. When this inequality holds,
$i^*=N_f$.

By substituting $v^*$ in (\ref{equ.proof_1.3}) into (\ref{equ.Lag_3}), the optimal caching distribution can be
derived as
\begin{equation} \label{equ.proof_1.5}
p^*_c(i) = \frac{\beta}{\lambda \pi r_{c}^2 N_f}\sum_{j=1}^{N_f}\ln(\frac{j}{i})+\frac{1}{N_f}.
\end{equation}
\item Case 2: When $i^*<N_f$, $p^*_c(i^*) = x_{i^*} - v >0$ and $x_{i^*+1} - v \leq 0$. By
    substituting $v$ in (\ref{equ.proof_1.3}) into these two inequalities, we have
\begin{equation} \label{equ.proof_1.6}
\begin{split}
\sum_{i=1}^{i^*}(x_i - x_{i^*+1}) \geq 1, \quad \sum_{i=1}^{i^*}(x_i - x_{i^*}) < 1
\end{split},
\end{equation}
which can be further derived by considering $p_r(i)$ in (\ref{equ.p_r}) and $x_i$ in \eqref{xi} as
\begin{equation} \label{equ.proof_1.7}
\begin{split}
\frac{\beta}{\lambda \pi r_{c}^2}\ln(\frac{(i^*+1)^{i^*}}{i^*!})\geq 1, \quad \frac{\beta}{\lambda \pi r_{c}^2}\ln(\frac{(i^*)^{i^*}}{i^*!}) < 1.
\end{split}
\end{equation}
Then, $i^*$ satisfies $\frac{(i^*+1)^{i^*}}{i^*!} \geq e^{\frac{\lambda \pi r_{c}^2}{\beta}}$ and
$\frac{(i^*)^{i^*}}{i^*!} < e^{\frac{\lambda \pi r_{c}^2}{\beta}}$. By
substituting  $v$ in (\ref{equ.proof_1.3}) into (\ref{equ.Lag_3}), we obtain
\begin{eqnarray}
\label{equ.proof_1.8}
p^*_c(i)=
\begin{cases}
\frac{\beta}{\lambda \pi r_{c}^2 i^*}\sum_{j=1}^{i^*}\ln(\frac{j}{i})+\frac{1}{i^*}, &i \leq i^* \\
0, &i^*<i \leq N_f.\\
\end{cases}
\end{eqnarray}

With Stirling formula \cite{SF.Approx}, $\sqrt{2\pi n}(\frac{n}{e})^n < n! < \sqrt{2\pi
n}(\frac{n}{e})^ne^{\frac{1}{12}}$, (\ref{equ.proof_1.7}) can be further derived  as
\begin{equation} \label{equ.proof_1.9}
\begin{split}
i^* - \ln(\sqrt{2 \pi i^*} e^{\frac{1}{12}}) < \frac{\lambda \pi r_{c}^2}{\beta}, i^* + 1 - \ln(\sqrt{2 \pi i^*}) > \frac{\lambda \pi r_{c}^2}{\beta}.\nonumber
\end{split}
\end{equation}

Since $i^*<N_f$, $\ln(\sqrt{2 \pi i^*}) < \ln(\sqrt{2 \pi N_f})$, we can obtain the range of $i^*$ as $\frac{\lambda
\pi r_{c}^2}{\beta} -1 \leq i^* \leq \frac{\lambda \pi r_{c}^2}{\beta} + \ln(\sqrt{2\pi N_f}) + 1$.
\end{itemize}

\section*{Appendix $B$: Proof of Proposition \ref{p:2}}
By denoting $z=P_t(r)$, $a = \frac{r^{-\alpha}}{\sigma_0^2}$, $b=\eta P_c$ and $A = \frac{F}{W\eta}$, the objective
function in (\ref{equ.opt2}) can be rewritten as
\begin{equation} \label{equ.proof_2.1}
f(z) = A\frac{z+b}{\log_2(1+az)},
\end{equation}
where $z,a,b>0$ and $z\leq P_{max}$. By taking the first-order derivative of $f(z)$, we obtain
\begin{equation} \label{equ.proof_2.2}
\begin{split}
&f'(z) = A\frac{((1+az)\log_2(1+az)\ln2-(1+az)-\epsilon)}{(1+az)\log_2^2(1+az)\ln2 }.
\end{split}
\end{equation}
Denote $g(y)=  y\log_2(y)\ln2-y-\epsilon$, where $y=1+az, 1<y\leq 1+aP_{max} \triangleq y_{0}$, and $\epsilon =
\frac{r^{-\alpha}\eta P_c}{\sigma_0^2}-1$. We can see that if $g(y)>0$, then $f'(z)>0$, and vice versa.

By taking the first-order derivative of $g(y)$, we obtain $g'(y) = \log_2(y)\ln2 > 0 $ due to $y>1$. Therefore, $g(y)$ is a
increasing function of $y$, when $y \rightarrow 1$, $g(y) \rightarrow 0$. If $g(y_{0}) \leq 0$, $g(y)$ is always less
than zero since $1<y<y_{0}$, then $f(z)$ will be a decreasing function, i.e., $f(z)$ achieves its minimum when
$z=P_{max}$. Otherwise, $f(z)$  first decreases and then increases, and  achieves its minimum when $g(y)=0$.

In summary, if $g(1+\frac{P_{max}r^{-\alpha}}{\sigma_0^2}) \leq 0$, $P^*_t(r) = P_{max}$. Otherwise, we can find
$P^*_t(r)$ efficiently by  searching from $(\frac{y}{e})^y = 2^{{\varepsilon}/{\ln2}}$, which is obtained from $g(y)=0$.
\vspace{-0.30cm}
\section*{Appendix $C$: Proof of Proposition \ref{p:3}}
Consider two collaboration distances $r_{c}'>r_{c}$, and the corresponding optimal caching distribution are respectively $p^*_c(i)'$ and
$p^*_c(i)$. We can see from (\ref{equ.p_o}) that $p_o(p^*_c(i),r)$ is an increasing function of $r$ $(0 < r \leq r_c)$.
  Then, $p_o(p^*_c(i)',r_c') \geq p_o(p^*_c(i),r_c') \geq p_o(p^*_c(i),r_c)$. Thus, $p_o(p^*_c(i),r_c)$
is an increasing function of $r_{c}$.

By subtracting $\bar{E}^*_c(r_c)$ from $\bar{E}^*_c(r_c')$, we have
\begin{equation} \label{equ.proof_3.1}
\begin{split}
&\bar{E}^*_c(r_c') - \bar{E}^*_c(r_c) \\
&= \int_{0}^{r_{c}'} E^*_c(r) p_o'(p^*_c(i)',r) dr - \int_{0}^{r_{c}} E^*_c(r) p_o'(p^*_c(i),r) dr\\
&\geq \int_{0}^{\xi} E^*_c(r)\Delta {p}_o dr + \int_{\xi}^{r_{c}'} E^*_c(r)\Delta {p}_odr\\
& \geq E^*_c(\xi) (\int_{\xi}^{r_{c}'} \Delta {p}_odr + \int_{0}^{\xi} \Delta {p}_odr)\\
& = E^*_c(\xi)( \int_{0}^{r_c'}  p_o'(p^*_c(i)',r) dr - \int_{0}^{r_c'}  p_o'(p^*_c(i),r) dr )\\
& = E^*_c(\xi)( p_o(p^*_c(i)',r)- p_o(p^*_c(i),r)) \geq 0.\nonumber
\end{split}
\end{equation}
where $\Delta{p}_o =  p_o'(p^*_c(i)',r) - p_o'(p^*_c(i),r)$, $0<\xi\leq r_c'$. The first inequality comes from  changing the
 integral upper limit  $r_c$ to $r_c'$. The second inequality is because $E^*_c(r)$ is an increasing function
of $r$ and by taking the partial derivative of $\Delta {p}_o \geq 0$ with respect to $r_c$. With tedious
derivation, we can show that there exist $\xi$, when $r_c<\xi$, $\Delta {p}_o \leq 0$, resulting in
$\int_{0}^{\xi} E^*_c(r)\Delta {p}_o dr \geq E^*_c(\xi)\int_{0}^{\xi} \Delta {p}_odr$, and when $r_c \geq \xi$,
$\Delta {p}_o \geq 0$, resulting in $\int_{\xi}^{r_{c}'} E^*_c(r)\Delta {p}_o dr \geq E^*_c(\xi)\int_{\xi}^{r_{c}'}
\Delta {p}_odr$. Thus, $\bar{E}_c(r_c)$ is an increasing function of $r_c$.

\bibliographystyle{IEEEtran}
\bibliography{CBQ_15}
\end{document}